\documentclass[10pt, conference, compsocconf]{IEEEtran}
\usepackage{colordvi,amssymb,amsmath,graphicx}
\usepackage{subfigure}

\newtheorem{theorem}{Theorem}[section]
\newtheorem{lemma}[theorem]{Lemma}
\newtheorem{corollary}[theorem]{Corollary}
\newtheorem{proposition}[theorem]{Proposition}
\newtheorem{definition}[theorem]{Definition}
\newtheorem{example}[theorem]{Example}
\newtheorem{proof}[theorem]{Proof}
\newtheorem{remark}[theorem]{Remark}

\newcommand{\ba}{{\bf a}}
\newcommand{\bb}{{\bf b}}
\newcommand{\bd}{{\bf d}}
\newcommand{\calA}{{\mathcal A}}

\newcommand{\R}{{\mathbb R}}

\newcommand{\Z}{{\mathbb Z}}

\newcommand{\demph}[1]{\emph{\Blue{#1}}}

\def\IEEEkeywordsname{Keywords}

\begin{document}
%
\title{Injectivity of 2D Toric B\'{e}zier Patches}


\author{\IEEEauthorblockN{Frank Sottile}
\IEEEauthorblockA{Department of Mathematics\\
Texas A\&M University\\
College Station, TX 77843, USA \\
Email: sottile@math.tamu.edu}
\and
\IEEEauthorblockN{Chun-Gang Zhu}
\IEEEauthorblockA{School of Mathematical Sciences\\
Dalian University of Technology\\
Dalian 116024, China\\
 Email: cgzhu@dlut.edu.cn}

}


\maketitle

\begin{abstract}
Rational B\'{e}zier functions are widely used as mapping functions in
surface reparameterization, finite element analysis, image warping
and morphing.
The injectivity (one-to-one property) of a mapping
function is typically necessary for these applications.
Toric B\'{e}zier patches are generalizations of classical patches
 (triangular, tensor product)
which are defined on the convex hull of a set of integer lattice
points. We give a geometric condition on the control
points that we show is equivalent to the injectivity of
every 2D toric B\'{e}zier patch with those control points
for all possible choices of weights.
This condition refines that of Craciun, et al.,
which only implied injectivity on the interior of a patch.

\end{abstract}



%
\IEEEpeerreviewmaketitle

\section{Introduction}
Mapping functions play
an important role in computer graphics, computer aided geometric
design (CAGD), finite element analysis (FEA) and some related areas.
The injectivity of mapping functions, that is, the absence of
self-intersection, is crucial in image warping and
morphing~\cite{Wol90}, free form deformation~\cite{Cho00}, surface
reparameterization, and so on.
Many authors have investigated conditions which imply injectivity.
Goodman and Unsworth~\cite{Goo96}
proposed a sufficient condition for the injectivity of a 2D B\'ezier
function.
For the control points of a $m\times n$ tensor product patch, their
condition involves $2m(m+1)+2n(n+1)$
linear inequalities. For image
morphing, Choi and Lee \cite{Cho00} presented a sufficient condition
for the injectivity of 2D and 3D uniform cubic B-spline functions.
Their condition provides a single bound for the displacements of
control points that guarantees the injectivity of the cubic B-spline
function.
Floater~\cite{Flo03} studies a sufficient condition for injectivity
of convex combination mappings over triangulations.

Fig.~\ref{fig:1} displays rational plane cubic B\'ezier curves with
their control polygons (bold lines). The curve in Fig.~\ref{fig:11}
has no points of self-intersection. The
 curve in Fig.~\ref{fig:12} has one point of self-intersection, which
may be removed by varying the weights  as shown in
Fig.~\ref{fig:13}. The control polygon of the first curve is in
convex position, so there are no positive weights for which the
resulting B\'ezier curve has self-intersection. For the other
control polygon there are weights (e.g.~Fig.~\ref{fig:12}) such that
the resulting B\'ezier curve has a point of self-intersection. The
cited works  provide conditions which imply no self-intersection.
Our purpose is different: We give conditions on the control points
for 2D patches which are equivalent to there being no
self-intersection for any choice of positive weights.

\begin{figure*}[htb]
\centerline {\subfigure[]{\label{fig:11}
\begin{picture}(115,74)
    \put(12,3){\includegraphics[height=65pt]{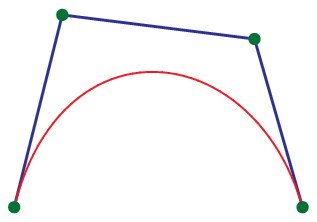}}
    \put(12,67){$\bb_1$}   \put(86,62){$\bb_2$}
    \put(0,0){$\bb_0$}    \put(105, 0){$\bb_3$}
   \end{picture}}\hfil
\subfigure[]{\label{fig:12}
\begin{picture}(115,74)
    \put(12,3){\includegraphics[height=65pt]{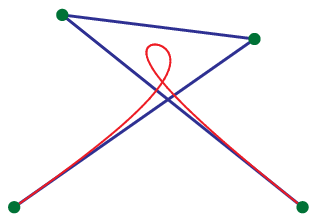}}
    \put(12,67){$\bb_2$}   \put(86,62){$\bb_1$}
    \put(0,0){$\bb_0$}    \put(105, 0){$\bb_3$}
   \end{picture}}\hfil
\subfigure[]{\label{fig:13}
\begin{picture}(115,74)
    \put(12,3){\includegraphics[height=65pt]{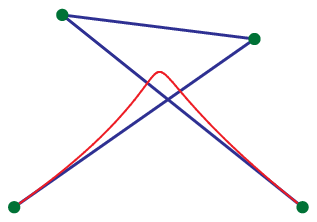}}
    \put(12,67){$\bb_2$}   \put(86,62){$\bb_1$}
    \put(0,0){$\bb_0$}    \put(105, 0){$\bb_3$}
   \end{picture}}}
\centering\caption{Cubic B\'ezier curves.}\label{fig:1}
\end{figure*}

The basic units in the geometric modeling of surfaces are
rational B\'ezier simplices and tensor product patches.
Krasauskas~\cite{KRA02} introduced toric B\'{e}zier
patches as a natural
extension of classical rational patches and their higher-dimensional generalizations, the
B\'{e}zier simploids by DeRose, et al.~\cite{DeRose}. The theory of
toric patches is based upon real toric varieties from algebraic
geometry \cite{Sot03}, and they provide a general framework in which to pose
many questions concerning classical rational patches.

To study dynamical systems arising from chemical reaction networks,
Craciun et al.~\cite{Cra05} prove an injectivity theorem for certain
maps. This was adapted in~\cite{CRA10} to give a geometric condition on a set of
control points which implies that the resulting toric B\'ezier patch
has no self-intersection, for any choice of positive weights.
That result
contains a minor flaw in that it only guarantees
injectivity in the interior of a patch.
We correct that flaw, at least for 2D patches, showing
that the condition from~\cite{CRA10} plus the mild
additional hypothesis that the vertices correspond to distinct control points
is equivalent to injectivity for every choice of positive weights.

In Section 2, we introduce toric B\'ezier patches as generalizations of the classical
rational patches.
In Section 3 we explain our condition and sketch its equivalence to the injectivity
of every 2D patch with a given set of control points, for all possible weights.
More details, including examples of the geometric arguments of Lemma~\ref{L:halfspaces}
and Corollaries~\ref{C:chull} and~\ref{C:edge} will be added in the complete version of
this paper.
We conclude some remarks on how to check this condition, argue that it is in
fact quite natural, and interpret it in terms of piecewise linear maps.

While our main interest is in establishing a criteria valid in
3D, and in fact in all dimensions,
we currently do not know how to add hypotheses to the
condition of~\cite{CRA10} so that the result will be equivalent to injectivity for any
choice of weights in 3D.

\section{Toric B\'{e}zier patches}

 Let $\calA\subset\Z^2$ be any finite set of integer \demph{lattice points}.
 Its \demph{convex hull}  $\Delta_\calA$ is a polygon
 whose vertices are lattice points.
 This polygon is also defined by its \demph{edge inequalities},
\[
  \Delta_\calA=\{(x,y)\in\R^2 \mid 0 \leq h_i(x,y)\,, i=1,\dotsc,\ell\},
\]
 where $h_i(x)=a_ix+b_iy+c_i$ are linear polynomials with integer
 coefficients and $(a_i,b_i)$ is relatively
 prime.

 For each integer lattice point $\ba\in{\calA}$, Krasauskas \cite{KRA02} defined the
\demph{toric Bernstein polynomial}
 \begin{equation}\label{eq:toric}
   \beta_\ba(x)\ :=\ h_1(x)^{h_1(\ba)}h_2(x)^{h_2(\ba)}\dotsb h_\ell(x)^{h_\ell(\ba)},
 \end{equation}
 These toric Bernstein polynomials are non-negative on $\Delta_{\calA}$, and the
 collection of all $\beta_\ba$ has no common zeroes in $\Delta_\calA$.

 Let $\R_{>}^\calA$ be $\R_{>}^{|\calA|}$ with coordinates
$(w_\ba\in\R_>\mid\ba\in\calA)$ indexed by elements of $\calA$.

\begin{definition}\label{def:toric}
 Let $\calA\subset\Z^2$ be a finite set.
 A toric B\'ezier patch associated with $\calA$ requires an assignment
 $f\colon\calA\to\R^d$  ($d=2,3$) of \demph{control points} and a choice of
 weights $w\in\R_>^\calA$.
 The \demph{toric B\'ezier patch} $F_w\colon\Delta_\calA\to \R^d$ is the function
 \begin{equation}\label{eq:tbp}
   F_w(x)\ =\ F_{\calA,f,w}(x)\ :=\
   \frac{\sum_{\ba\in{\calA}} w_\ba f(\ba)\beta_\ba(x) }
        {\sum_{\ba\in{\calA}} w_\ba\beta_\ba(x)}\,,
 \end{equation}
 written $F_w$ as $\calA$ and $f$ are understood.
\end{definition}

The degree of a toric B\'ezier patch is encoded in its domain, differing from the
classical patches as developed in~\cite{farin}.
These two types of patches share many properties,
which is explained in~\cite{KRA02,Sot03}.
Two properties in particular are important for us.

One is the convex hull property, that the image of $\Delta_\calA$
under $F_w$ is contained in the convex hull of the control points
$f(\calA)$ with $F_w(\bb)=f(\bb)$ if $\bb$ is a vertex of
$\Delta_\calA$, and the other is the boundary property, that the
restriction of $F_w$ to an edge $\delta$ of $\Delta_\calA$ is a
rational B\'ezier curve, defined by control points and weights
corresponding to lattice points of $\delta$.

The boundary property may be seen directly by considering
the restriction to an edge.
For the convex hull property, note that as $w_\ba\beta_{\ba}(x)$ is nonnegative,
$F_w(x)$ is a convex combination of the control points,
and if $\bb$ is a vertex, then $\beta_\ba(\bb)$ is zero unless $\ba=\bb$.
Since the toric Bernstein polynomials are strictly positive on the interior of $\Delta$
(and those corresponding to an edge $\delta$ are strictly positive on the interior of
$\delta$), we may deduce a little more.

\begin{proposition}
 The image of the interior of $\Delta$ lies strictly in the interior of the convex hull of
 the control points $f(\calA)$, and the image of the interior of an edge $\delta$ lies
 strictly within the interior of the convex hull of $f(\delta\cap\calA)$.
\end{proposition}

Toric B\'ezier patches include the classical B\'ezier patches and
some multi-sided patches such as Warren's
polygonal surface \cite{War92} which is a reparameterized toric
B\'ezier surface.

\begin{example}[Tensor product patches]\label{ex:tensor}
Let $m,n$ be positive integers. Let $\calA$ be the integer points in
the $m\times n$ rectangle $\calA:=\{(i,j) :  0\leq i\leq m,\
0\leq j\leq n\}$. Then the corresponding toric
Bernstein polynomials~\eqref{eq:toric} are
 \begin{equation}
   \beta_{(i,j)}(x,y)\ :=\   x^i(m-x)^{m-i} y^j(n-y)^{n-j}\,,
 \end{equation}
and the toric B\'ezier patch \eqref{eq:tbp} (with weights
$w_{i,j}=\binom{m}{i}\binom{n}{j}$) is the rational
tensor product B\'ezier patch of bidegree $(m,n)$ after the
simple reparameterization $s=x/m$, $t=y/n$.
\end{example}

\begin{example}[Triangular B\'ezier patches]\label{ex:triangle}
Let $m$ be a positive integer and $\calA$ be the integer points in the
triangle with vertices $(0,0)$, $(m,0)$, and $(0,m)$,
$\calA:=\{(i,j)\mid 0\leq i,j, \quad 0\leq  m-i,j\}$.
The corresponding Bernstein polynomials~\eqref{eq:toric} are
\[
   \beta_{i,j}(x,y)\ =\  x^i y^j (m-x-y)^{m-i-j}.
\]
Then the toric B\'ezier patch~\eqref{eq:tbp} (with weights
$w_{i,j}=\frac{m!}{i!j!(m-i-j)!}$) is the rational
B\'ezier triangle of degree $m$ after the simple reparameterization $s=x/m$, $t=y/m$.
\end{example}

\section{Injectivity of 2D toric B\'{e}zier patches}
Given a finite set $\calA\subset\Z^2$  and a choice
$f\colon\calA\to\R^2$ of control points,
we consider the injectivity of toric
B\'ezier patches as mapping functions
$F_w\colon \Delta_{\calA}\mapsto \R^2$~\eqref{eq:tbp}, for all choices
$w\in\R^\calA_>$ of positive weights.

Affinely independent  points $\ba_0,\ba_1,\ba_2$
determine an orientation via the ordered basis
$\ba_1{-}\ba_0,\ba_2{-}\ba_0$ of $\R^2$.

\begin{definition}\label{def:co}
 A choice $f\colon\calA\to\R^2$ of control points is \demph{weakly compatible} if
\begin{enumerate}
 \item There are affinely independent points $\ba_0,\ba_1,\ba_2$ of $\calA$ such that
      $f(\ba_0),f(\ba_1),f(\ba_2)$ is also affinely independent, and
 \item For any affinely independent points $\ba'_0,\ba'_1,\ba'_2$ of $\calA$ with the same
   orientation as $\ba_0,\ba_1,\ba_2$, if
   $f(\ba'_0),f(\ba'_1),f(\ba'_2)$ is also affinely independent, then it has the same
   orientation as $f(\ba_0),f(\ba_1),f(\ba_2)$.
 \end{enumerate}
\end{definition}

Fig.~\ref{fig:2} shows three sets of labeled points, indicating
assignments between them.
 \begin{figure}[htb]
  \begin{picture}(242,56)(-6,0)
   \put(0,2){\includegraphics{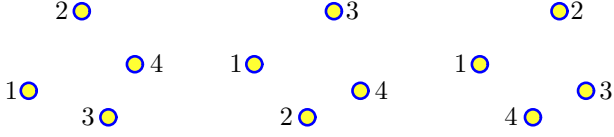}}
   \put(  1,17){$1$}   \put( 20,47){$2$}
   \put( 30, 7){$3$}   \put( 56,27){$4$}

   \put( 86,27){$1$}   \put(130,47){$3$}
   \put(105, 7){$2$}   \put(141,17){$4$}

   \put(171,27){$1$}   \put(215,47){$2$}
   \put(190, 7){$4$}   \put(226,17){$3$}
  \end{picture}
 \caption{Weak compatibility.}\label{fig:2}
\end{figure}
The assignment between the first two sets
is weakly compatible, but neither assignment to the third set is
weakly compatible.

We state Theorem~3.5 of~\cite{CRA10} for $\R^2$, which is their main
result on injectivity of toric B\'{e}zier functions (it holds in any
dimension). Write $\Delta^{\circ}_{\calA}$ for the interior of
$\Delta_{\calA}$.

\begin{theorem}\label{thm:clf-injectivity1}
 The map $F_w:\Delta^{\circ}_{\calA}\mapsto \R^2$ is injective for all
 $w\in\R^{\calA}_>$ if and only if the assignment $f\colon\calA\to\R^2$ is
 weakly compatible.
\end{theorem}

In~\cite{CRA10}, the authors
incorrectly stated this result as $F_w$
is injective on all of $\Delta_\calA$, even though their proof was only valid for the
interior of the convex hull.
Their proof showed that $F_w$ has no critical points in the interior,
which shows that it is an open map  on $\Delta_\calA^\circ$.

This is the best possible result with these hypotheses:
Consider a bilinear patch where two control points coincide.
Specifically, let $\calA=\{(0,0),(0,1),(1,0),(1,1)\}$ and suppose that the control points
are $\{(0,0),(0,1),(1,0)\}$, where $f(\ba)=\ba$, except that $f(1,1)=(1,0)$.
This assignment of control points is weakly compatible, but $F_w$ collapses
the edge between $(1,0)$ and $(1,1)$ to the point $(1,0)$.
\[
   \begin{picture}(240,54)(0,-5)
    \put(28,0){\includegraphics{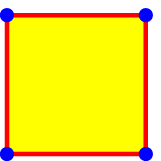}}
      \put(  0,-5){$(0,0)$}\put(  0,40){$(0,1)$}
      \put( 74,-5){$(1,0)$}\put( 74,40){$(1,1)$}
     \put(105,15){$\xrightarrow{\ F_w\ }$}
     \put(168,0){\includegraphics{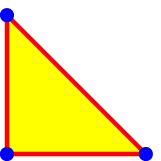}}
      \put(140,-5){$(0,0)$}\put(140,40){$(0,1)$}
      \put(214,-5){$(1,0)$}
   \end{picture}
\]
This example shows that more hypotheses are needed to ensure that $F_w$ is
injective on $\Delta_\calA$, and those hypotheses should imply that faces of
$\Delta_\calA$ are not collapsed.
In fact, this is the only additional hypothesis needed.

\begin{definition}\label{def:stco}
 A choice $f\colon\calA\to\R^2$ of control points is \demph{compatible} if
 it is weakly compatible, and no two vertices have the same image under $f$.
\end{definition}

We state our main result.

\begin{theorem} \label{T:equivalence}
 The map $F_w:\Delta_{\calA}\mapsto \R^2$ is injective for all
 $w\in\R^{\calA}_>$ if and only if the assignment $f\colon\calA\to\R^2$ is
 compatible.
\end{theorem}

If $\ba\in\calA$ is a vertex of $\Delta_\calA$, then $F_w(\ba)=f(\ba)$.
Theorem~\ref{thm:clf-injectivity1}, together with this observation,
shows that if $F_w$ is injective for all $w\in\R^\calA_>$, then $f\colon\calA\to\R^2$ is
compatible.

For the other implication, suppose that $f\colon\calA\to\R^2$ is compatible.
We show that the assumption that $F_w$ is not injective leads to a contradiction.

We first make several observations about the relative positions of the points $f(\ba)$
for $\ba\in\calA$ which are implied by compatibility.
Composing with a reflection of $\R^2$ if necessary, we may assume that  if
$\ba_0,\ba_1,\ba_2$ and $f(\ba_0),f(\ba_1),f(\ba_2)$ are both affinely independent, then
they induce the same orientation on $\R^2$.

Let $\delta$ be an edge of $\Delta_\calA$.
There is some triple of points $\bd,\bd',\ba$ of $\calA$ with $f(\bd),f(\bd'),f(\ba)$
affinely independent where $\bd,\bd'\in\delta$ and $\ba\not\in\delta$.
Indeed, if there are no such triples, then every point of $f(\calA)$ lies on every line
segment between two distinct points of $f(\delta\cap\calA)$, which implies that the points
of $f(\delta\cap\calA)$ are collinear and the line they span contains $f(\calA)$, which
contradicts the first condition for weak compatibility of Definition~\ref{def:co}.
This argument requires that there be at least two distinct points of $f(\delta\cap\calA)$,
which follows as the endpoints of $\delta$ (which are vertices of $\Delta_\calA$) are
mapped to different points under $f$.

Suppose that we list the points $\bd_0,\bd_1,\dotsc,\bd_m$
of $\delta\cap\calA$ so that if $\ba\in\calA\setminus\delta$, and $i<j$, then
$\bd_i,\bd_j,\ba$ are positively oriented.
Then either $f(\bd_i),f(\bd_j),f(\ba)$ are collinear or positively oriented.
Since there must be at least one such triple with $f(\bd_i),f(\bd_j),f(\ba)$ affinely
independent, we deduce the following.

\begin{lemma}\label{L:halfspaces}
 Every control point $f(\calA\setminus\delta)$ lies in the intersection of closed
halfspaces
\[
   \overline{\{x\in\R^2\mid f(\bd_i),f(\bd_j),x\mbox{ are positively oriented}\}}
\]
for $i<j$ with $f(\bd_i)\neq f(\bd_j)$, and this intersection has a nonempty relative
interior.
\end{lemma}

\begin{corollary}\label{C:chull}
 For every edge $\delta$ of $\Delta_\calA$ and every $\bb\in\calA\setminus\delta$, the
 control point $f(\bb)$ does not lie in the relative interior of the convex hull of
 $f(\delta\cap\calA)$.
\end{corollary}

To see this, note that the intersection of halfspaces of Lemma~\ref{L:halfspaces} is
either interior or exterior to the convex hull of $f(\delta\cap\calA)$, and if it is
exterior, then it is separated from the relative interior of the convex hull by a line.
If there is an edge $\delta$ so that this intersection lies in the interior of the convex
hull of  $f(\delta\cap\calA)$, let $\delta'$ be a different edge.
Then the positions of the points of $\delta\cap\calA$ relative to the intersection of
halfspaces for $\delta'$ leads to a contradiction.

\begin{corollary}\label{C:edge}
 If $f\colon\calA\to\R^2$ is compatible, then the restriction of $F_w$ to any edge
 $\delta$ of $\Delta_\calA$ is injective.
\end{corollary}

To see this, fix an edge $\delta$ and consider the intersection of halfspaces of
Lemma~\ref{L:halfspaces}.
This intersection is exterior to the convex hull of $f(\delta\cap\calA)$ and so consists
of an unbounded polyhedron, $P$.
Consider the orthogonal projection $\pi\colon\R^2\to\R$ along an unbounded direction of $P$.
Then the map $\pi\circ f\colon \delta\cap\calA\to \R$ is a weakly compatible choice of
control points for $\delta\cap\calA$, and so the map $\pi\circ F_w$ restricted to the edge
$\delta$ is injective, by Theorem~\ref{thm:clf-injectivity1}.
But this implies that the restriction of $F_w$ to $\delta$ is injective.

\begin{proof}[Proof of Theorem~\ref{T:equivalence}]
 We suppose that $f\colon\calA\to\R^2$ is compatible and that $F_w$ is not injective.
 Let $x,y\in\Delta_\calA$ be distinct points with $F_w(x)=F_w(y)$.

 First, neither $x$ nor $y$ can be a point of $\Delta^\circ_\calA$.
 To see this, suppose that $x\in\Delta^\circ_\calA$ and let $V$ be a neighborhood of $x$
 in $\Delta_\calA$ whose closure does not contain $y$.
 Then $F_w(V)$ is an open set containing $F_w(x)=F_w(y)$, so
 $F_w^{-1}(V)\setminus V$ contains an open subset $U$ of $y$ in $\Delta_\calA$.
 But then points of $U\cap\Delta_\calA^\circ$ are mapped by $F_w$ to points of
 $F_w(V)$, and so $F_w$ is not injective on the interior of $\Delta$, which contradicts
 Theorem~\ref{thm:clf-injectivity1}, as the choice $f$ of control points is weakly compatible.

 Thus $x$ and $y$ are points of some edges of $\Delta_\calA$.
 They cannot be points of the same edge $\delta$, for then the restriction of $F_w$ to
 $\delta$ is not injective, contradicting Corollary~\ref{C:edge}.
 Thus they are points of different edges, $x\in\delta$ and $y\in\delta'$
 with $\delta\neq\delta'$.
 We cannot have one of them be an interior point of its edge, for then
 the relative interiors of the convex hulls of
 $f(\delta\cap\calA)$ and $f(\delta'\cap\calA)$ meet, contradicting
 Corollary~\ref{C:chull}.

 The only possibility left is that $x$ and $y$ are vertices of $\Delta$, but then
 $F_w(x)=f(x)$ and $F_w(y)=f(y)$, which are different, as the choice $f$ was compatible.
\end{proof}

\begin{remark}

   By definition, to check weak compatibility for 2D patches, it suffices to check
   determinants for each triple of points of $\calA$
   and the corresponding control points, giving a
   simple $(\#(\calA))^3$ algorithm.
   The complexity may be reduced if we start from a triangulation of $\Delta_\calA$, or
   with careful bookkeeping.
   Such triangulations can be obtained from control nets for tensor
   product patches or B\'ezier triangles.
   We will treat the complexity of checking weak compatibility in the complete version of
   this extended abstract.

 Mapping functions that are weakly compatible exist; for example the identity
  assignment of control points is weakly compatible.
  A designer may choose weakly
  compatible control points for aesthetic or other reasons.
  For example, if only a few control points are moved such as in
  image warping, morphing, or reparameterization, then the control points may be
  weakly compatible by design, or else only a few determinants need to be computed.

 For any triangulation of $\calA$,  the assignment of control points induces a piecewise
 linear map to the image.
 This piecewise linear map is injective (except possibly collapsing
 an interior simplex) for every such triangulation if and only if the assignment of
 control points is weakly compatible.
\end{remark}

\section{Conclusions}
In this paper, we study the injectivity of toric B\'ezier patch
geometrically. We present a simple condition on a set of control
points which implies that the resulting 2D toric B\'ezier patch is
injective, for any choice of positive weights. For higher dimension,
the best result remains Theorem~\ref{thm:clf-injectivity1} by
Craciun et al. in \cite{CRA10} (Theorem~3.5 in \cite{CRA10}). We
plan to continue this investigation of injectivity for 3D and higher
dimensions in a future publication.

\section*{Acknowledgements}
 The authors thanks to Tim Goodman and Keith Unsworth for providing
their paper \cite{Goo96}. Research of Sottile is supported in part
by NSF grant DMS-1001615, and the Institut Mittag-Leffler,
Djursholm, Sweden. Research of Zhu is supported by the NSFC (Grant
Nos. 10801024, 11071031, and U0935004), the Fundamental Research
Funds for the Central Universities (DUT10ZD112, DUT11LK34), and the
National Engineering Research Center of Digital Life, Guangzhou
510006, China.


\begin{thebibliography}{1}

\bibitem{Cho00}Y. Choi and S. Lee,  ``Injectivity conditions of 2D and 3D uniform cubic B-Spline functions," Graphical Models, vol. 62, 2000, pp. 411-427.

\bibitem{CRA10}G. Craciun, L. Garc\'ia-Puente and F. Sottile, ``Some geometrical aspects
  of control points for toric patches," in: Mathematical Methods for Curves and Surfaces,
  M. D\"ahlen
et al. Eds,  Lecture Notes in Computer Science, vol. 5862, Springer,
Heidelberg, 2008, pp. 111-135.

\bibitem{Cra05}G. Craciun and M. Feinberg, ``Multiple equilibria in complex
  chemical reaction networks I. The injectivity property," SIAM Journal on
  Applied Mathemathics, vol. 65, 2005, pp. 1526-1546.

\bibitem{DeRose}T. DeRose, R. Goldman, H. Hagen, and S. Mann, ``Functional composition algorithms via blossoming," ACM Transactions on
Graphics, vol. 12, 1993, pp. 113¨C135.

\bibitem{farin}G. Farin, Curves and Surfaces for Computer Aided Geometric Design, Computer Science and
Scientific Computing, San Diego: Academic Press Inc., 1997.

\bibitem{Flo03}M. Floater, ``One-to-one piecewise linear mappings over triangulations,"
Mathematics of Computation, vol. 72, 2003, pp. 685-696.

\bibitem{Goo96}T. Goodman and K. Unsworth, ``Injective bivariate maps," Annals of Numerical Mathematics, vol. 3, 1996, pp. 91-104.

\bibitem{KRA02}R. Krasauskas,  ``Toric surface patches," Advances in Computational
  Mathematics, vol. 17, 2002, pp. 89-113.


\bibitem{Sot03}F. Sottile, ``Toric ideals, real toric varieties, and the moment
map," in: Topics in Algebraic Geometry and Geometric Modeling,
Contemp. Math., vol. 334, Amer. Math. Soc., Providence, RI, 2003,
pp. 225-240.

\bibitem{War92}J.  Warren, ``Creating multisided rational B\'ezier surfaces using base points," ACM Transactions on Graphics, vol. 11, 1992, pp. 127-139.

\bibitem{Wol90}G. Wolberg, Digital Image Warping, Los Alamitos: IEEE Computer Society Press, CA, 1990.

\end{thebibliography}
\end{document}